\documentclass[10pt,letterpaper,conference]{ieeeconf}

\IEEEoverridecommandlockouts         
\usepackage{graphicx} 
\usepackage{amsmath} 
\usepackage{amssymb}  
\usepackage{amsfonts}
\usepackage{cite}
\usepackage{xcolor}
\usepackage{epstopdf}

\newtheorem{lemma}{Lemma}
\newtheorem{theorem}{Theorem}

\newtheorem{remark}{Remark}

\title{\LARGE \bf Pole placement design for quantum systems via coherent observers
\thanks{This research was supported by the Australian Research Council Centre of Excellence for Quantum Computation and Communication Technology (project number CE110001027), AFOSR Grant FA2386-12-1-4075, and the Australian Research Council Discovery Project program (Discovery project DP140101779).}}

\author{Zibo Miao,~Matthew~R.~James and Valery A. Ugrinovskii
\thanks{Z. Miao is with the Department of Electrical \& Electronic Engineering, The University of Melbourne, Parkville, Victoria 3010, Australia (email: Zibo.Miao@unimelb.edu.au).}
\thanks{M. R. James is with the ARC Centre for Quantum Computation and Communication Technology, Research School of Engineering, The Australian national University, Canberra ACT 2601, Australia (email:  Matthew.James@anu.edu.au).}
\thanks{V. A. Ugrinovskii is with the School of Engineering and Information Technology, University of New South Wales Canberra at the Australian Defence Force Academy, Canberra ACT 2600, Australia  (email: V.Ougrinovski@adfa.edu.au).}
}

\begin{document}

\maketitle

%


\begin{abstract}
We previously extended Luenberger's approach for observer design to the quantum case, and developed  a class of coherent observers which tracks linear quantum stochastic systems in the sense of mean values. In light of the fact that the Luenberger observer is commonly and successfully applied in classical control, it is interesting to investigate the role of coherent observers in quantum feedback. As the first step in exploring observer-based coherent control, in this paper we study pole-placement techniques for quantum systems using coherent observers, and in such a fashion, poles of a closed-loop quantum system can be relocated at any desired locations. In comparison to classical feedback control design incorporating the Luenberger observer, here direct coupling between a quantum plant and the observer-based controller are allowed to enable a greater degree of freedom for the design of controller parameters. A separation principle is presented, and we show how to design the observer and feedback independently to be consistent with the laws of quantum mechanics. The proposed scheme is applicable to coherent feedback control of quantum systems, especially when the transient dynamic response is of interest, and this issue is illustrated in an example. 
\end{abstract}

\section{Introduction}
The last two decades have witnessed evident growth in quantum technologies, with control of quantum systems being the main focus of quantum engineering \cite{WM10, Mabuchi05, DP10, PZMGUJ14, Hush:2013}. It is increasingly apparent that coherent feedback control may have significant advantages over measurement-based feedback \cite{SL00, HM08, HM12, NY14}. One of the main reasons is that, measurements of quantum systems tend to be noisy ascribed to macroscopic read-out devices, and measurement based controllers are classical systems that cannot be realised by quantum hardwares \cite{Hush:2013a}. By contrast, a coherent feedback controller can be designed on the same scale as the quantum plant. On the other hand, measuring a quantum system can cause irreversible damage due to  the nature of quantum mechanics \cite{WM10}. However, a coherent feedback controller directly makes use of non-commutative quantum signals, without the need for a measurement device, and therefore this scheme retains the coherence in the whole process \cite{WM94b, HM12, KAKKCSL13}.

Although a remarkable development of coherent feedback control has been achieved recently, this framework is still in its incipient stage. For that matter, it lacks many tools available in classical or measurement-based control. It is well established that estimating a classical linear dynamic system from a series of noisy measurements using the Kalman filter can provide improved performance over direct feedback schemes \cite{AM79,AM89}. Unfortunately, traditional techniques do not appear to be applicable to coherent feedback due to difficulties with quantum conditioning onto non-commutative subspaces of signals  \cite{AMPJ15,VPB93,BHJ07, NJP09, VA14}. Furthermore, due to the Heisenberg uncertainty principle, there is a fundamental limit for the mean squared estimation error. 

The closest option instead of least squares estimators for quantum systems is so called coherent observers. A coherent observer, as a quantum system itself described by quantum stochastic differential equations (QSDEs) in the Heisenberg picture, is driven by the output field of the quantum plant without measurements involved. In our previous work, we extended Luenberger's approach for observer design to the quantum case, and developed a class of coherent observers that can track quantum stochastic systems asymptotically in the sense of mean values \cite{MJ12, MEPUJ13, MHJ15J}.  Observers are already of importance in classical control \cite{DL66,GE02}.
For example, in classical control theory, the pole-placement problem has attracted enormous interest with some systematical approaches developed \cite{ZDG95, KO10}. It is widely shared that, in many practical situations, only a limited number of plant variables to be controlled are available for measurement. Fortunately, these plant variables can be reconstructed by the Luenberger observer, and the observer may be incorporated in classical feedback control design \cite{DL64,DL66}.  Hence, we expect coherent quantum observers will have similar utility in quantum feedback control if it is designed so that it possesses the property that the Luenberger observer has.  In particular, if we desire an observable (a self-adjoint operator defined on a Hilbert space to represent physical quantities in quantum mechanics) to be asymptotically stable in the mean with specific transient response, a mean tracking coherent observer can be employed to provide a reliable estimate.
In addition, a coherent observer and the plant are correlated in the sense that some quantum features can be observed in the joint system \cite{MJ12,MHJ15J}. As the first step in studying observer-based coherent feedback control design, in this paper we are concerned with the pole-placement technique using coherent observers. The pole-placement technique is essential in the quantum case as it is related to transient responses of quantum systems, which is an important topic in quantum engineering, and certain quantumness can be made salient by manipulating the closed-loop poles \cite{HRHJ13, GNW10, KB06 }. The observer-based feedback control protocol is shown in Fig. \ref{fig:obs_con}, and all the information flows are at the quantum level.
As opposed to the classical case, a quantum output feedback controller of the following form
\begin{align*}
u = -\; \text{gain} \; y + \text{noise}
\end{align*}
cannot be physically realised \cite{JNP08, NJP09}. For that reason, a coherent observer plays a vital role in coherent feedback control design aimed at reassigning the poles. Furthermore, we allow for direct coupling (represented by the black double headed arrow in Fig. \ref{fig:obs_con}) between a linear quantum plant and the corresponding coherent observer \cite{IRP14, ZJ11}, and thus a greater degree of freedom for the controller design can be obtained. A separation principle is provided in this scenario, and we show how to design the observer and feedback independently consistent with the laws of quantum mechanics. The proposed approach in this paper can be applied to a variety of quantum systems (e.g., optomechanical systems), with the purpose of tuning parameters and improving the transient response performance.
\begin{figure}[!htp]
\centering
\includegraphics[scale=.5]{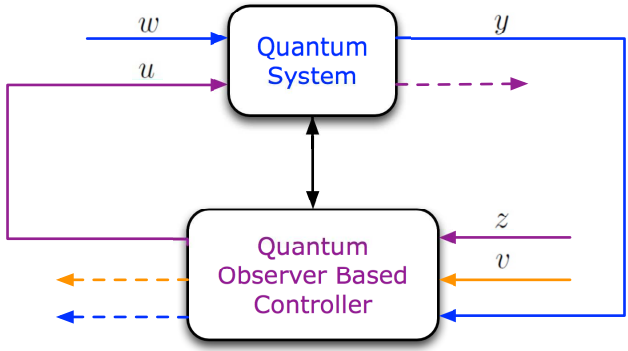}
\caption{Coherent quantum observer-based feedback control design.}
\label{fig:obs_con}
\end{figure}

The paper is organised as follows. We begin in Section \ref{sec:LQSS} by presenting linear quantum stochastic systems as quantum plants considered in this paper. In Section \ref{sec:OFCATPT}, we propose an observer-based approach to the pole-placement problem in a coherent fashion. This is followed by a specific example in Section \ref{sec:IE}, which illustrates the design and performance of observer-based feedback control. Section \ref{sec:C} provides some concluding remarks and future research directions.

{\em Notations}. In this paper $\ast$ is used to indicate the Hilbert space adjoint $X^{\ast}$ of an operator $X$, as well as the complex conjugate $z^{\ast}=x-iy$ of a complex number $z=x+iy$ (here, $i=\sqrt{-1}$ and $x,y$ are real). Real and imaginary parts are denoted by $\Re\left(z\right)=\frac{z+z^{\ast}}{2}$ and $\Im\left(z\right)=\frac{z-z^{\ast}}{2i}$ respectively. The conjugate transpose $A^\dagger$ of a matrix $A=\left\{ a_{ij}\right\} $ is defined by $A^{\dagger}=\left\{ a_{ji}^{\ast}\right\} $. Also defined are the conjugate  $A^{\sharp}=\left\{ a_{ij}^{\ast}\right\} $ and the transpose $A^{T}=\left\{ a_{ji}\right\} $ matrices, so that $A^{\dagger}=\left(A^T\right)^{\sharp}=\left(A^{\sharp}\right)^T$. $\textrm{det}\left(A\right)$ denotes the determinant of a matrix $A$, and $\textrm{tr}\left(A\right)$ represents the trace of $A$. The mean value (quantum expectation) of an operator $X$ in the state $\rho$ is denoted by $\left\langle X \right\rangle = \mathbb{E}_{\rho}\left[X\right]=\mathrm{tr}\left(\rho X\right)$. The commutator of two operators $X,Y$ is defined by $\left[X,Y\right]=XY-YX$. The anticommutator of two operators $X,Y$ is defined by $\left\{X,Y\right\}=XY+YX$. The tensor product of operators $X,Y$ defined on Hilbert spaces $\mathbb{H},\mathbb{G}$ is denoted $X \otimes Y$, and is defined on the tensor product Hilbert space $\mathbb{H}\otimes\mathbb{G}$. $I_n$ ($n\in\mathbb{N}$) denotes the $n$ dimensional identity matrix. $0_n$ ($n\in\mathbb{N}$) denotes the $n$ dimensional zero matrix.

\section{Linear quantum stochastic systems}
\label{sec:LQSS}

The dynamics of an open quantum system are uniquely determined by the triple $(S,L,\mathcal{H})$ \cite{KP92,GJ09J1}. The self-adjoint operator $\mathcal{H}$ is the Hamiltonian describing the self-energy of the system. The unitary matrix $S$ is a scattering matrix, and the column vector $L$  with operator entries is a coupling vector. $S$ and $L$ together specify the interface between the system and the fields. We assume there is no interaction between different fields, and thus hereafter we assume $S$ to be the identity matrix without loss of generality \cite{JNP08, NJP09, GZ00}.

Given an operator $X$ defined on the initial Hilbert space $\mathbb{H}$, the Heisenberg evolution is defined by
\begin{align}
dX = & \left(\mathcal{L}\left(X\right)-i\left[X,\mathcal{H}\right]\right)dt \nonumber\\
&+ \frac{1}{2}\left([X,L]-[X,L^\dagger]
\right)dW_1  \nonumber \\
&-\frac{i}{2}\left([X,L]+[X,L^\dagger]
\right)dW_2,
\label{eq:SLHqua}
\end{align}
with
\begin{align*}
\left[\begin{array}{c} W_1 \\  W_2 \end{array}\right] = \left[\begin{array}{c} W+W^\sharp  \\ -i(W-W^\sharp) \end{array}\right],
\end{align*}
and 
\begin{align}
\mathcal{L}\left(X\right)=\frac{1}{2}L^{\dagger}\left[X,L\right]+\frac{1}{2}\left[L^{\dagger},X\right]L,
\label{eq:Lindblad}
\end{align}
which is called the Lindblad operator. The operators $W$ are defined on a particular Hilbert space called a Fock space $\mathbb{F}$.
Input field quadratures $W+W^{\sharp}$ and $-i\left(W-W^{\sharp}\right)$ are each equivalent to classical Wiener processes, but do not commute. A field quadrature can be measured using homodyne detection \cite{GZ00,GJ09J1}.  In addition, the quadrature form of the output fields is given by
\begin{align}
\left[\begin{array}{c} dY_1 \\  dY_2 \end{array}\right] = \left[\begin{array}{c} L+L^\sharp  \\ -i(L-L^\sharp) \end{array}\right] \,dt
+ \left[\begin{array}{c} dW_1 \\ dW_2 \end{array}\right].
\label{eq:outfields}
\end{align}

In this work we focus on open harmonic oscillators as quantum plants. The dynamics of each oscillator are described by two self-adjoint operators position $q_j$ and momentum $p_k$, which satisfy the canonical commutation relations $[q_j,p_k] = 2i \delta_{jk}$ where $\delta_{jk}$ is the Kronecker delta \cite{Nielsen:2010}. It is convenient to collect the position and momentum operators of the oscillators into an $n_x$-dimensional column vector $x\left(t\right)$, defined by  $x\left(t\right)=\left(q_{1}\left(t\right),p_{1}\left(t\right),q_{2}\left(t\right),p_{2}\left(t\right), \ldots,q_{\frac{n_x}{2}}\left(t\right),p_{\frac{n_x}{2}}\left(t\right)\right)^{T}$. 
In this case the commutation relations can be re-written as:
 \begin{equation}
 x\left(t\right)x\left(t\right)^{T}-\left(x\left(t\right)x\left(t\right)^{T}\right)^{T}=2i\Theta_{n_x} \label{cancom}
 \end{equation}
 where $\Theta_{n_x} = I_{\frac{n_x}{2}}\otimes J$ with $J=\left[\begin{array}{cc}0 & 1\\-1 & 0\end{array}\right]$. In general, $\Theta_{n} = I_{\frac{n}{2}}\otimes J$ for any even number $n\in\mathbb{N}$.
  
Linear quantum plants, as given below, in particular, are defined by having a quadratic Hamiltonian of the form $\mathcal{H}_p = \frac{1}{2}x^{T}R_px$ with $R_p$ being a $\mathbb{R}^{n_x \times n_x}$ symmetric matrix; coupling operators of the form $L_w=\Lambda_wx$ with $\Lambda_w$ being a $\mathbb{C}^{\frac{n_w}{2} \times n_x}$ matrix and $L_u=\Lambda_ux$ with $\Lambda_u$ being a $\mathbb{C}^{\frac{n_u}{2} \times n_x}$ matrix (here $n_x$, $n_w$, $n_u$ and $n_y$ are positive even numbers). If we use an $n_y$-dimensional column vector $y\left(t\right)$ to incorporate all the quadratures of the output fields then, based on \eqref{eq:SLHqua} and \eqref{eq:outfields}, the dynamics of a set of open harmonic oscillators can be described by the following linear QSDEs \cite{JNP08}:
\begin{subequations}
\label{eq:qsys}
\begin{align}
dx\left(t\right)&=Ax\left(t\right)dt+B_1dw\left(t\right) + B_2du\left(t\right),\\
dy\left(t\right)&=Cx\left(t\right)dt+dw\left(t\right),
\end{align}
\end{subequations}
where $A$, $B_1$, $B_2$, $C$ are $\mathbb{R}^{n_{x}\times n_{x}}$, $\mathbb{R}^{n_{x}\times n_{w}}$, $\mathbb{R}^{n_{x}\times n_{u}}$ and $\mathbb{R}^{n_{y}\times n_{x}}$ matrices respectively. The $n_u$-dimensional column vector  $u\left(t\right)$ denotes the quantum signal fed back to the plant from the observer-based controller, which we will explain in detail later. The coefficient matrices $A$, $B_1$, $B_2$ and $C$ satisfy the following physical realisability conditions (e.g., see \cite{JNP08}):
\begin{subequations}
\label{eq:qsyscommu}
\begin{align}
&A\Theta_{n_{x}}+\Theta_{n_{x}} A^{T}+B_1\Theta_{n_{w}} B_1^{T} + B_2\Theta_{n_{u}} B_2^{T}=0,\\
&B_1=\Theta_{n_{x}} C^{T}\Theta_{n_{w}}.
\end{align}
\end{subequations}
These algebraic constraints on the coefficient matrices of (\ref{eq:qsys}) were originally derived by requiring the communication relations hold for all times, a property enjoyed by open physical systems undergoing an overall unitary evolution \cite{JNP08, GZ00}.

\section{Observer-based feedback controller and the pole-placement technique}
\label{sec:OFCATPT}

A coherent observer is another system of quantum harmonic oscillators which we engineer such that at least the system variables track those of the quantum plant asymptotically in the sense of mean values \cite{MJ12, VA14, MHJ15J}. In classical control theory, it is well known that if not all state variables of a linear plant are available for feedback, an observer may be needed for feedback design \cite{ZDG95, GE02}.  In this section we present a quantum coherent counterpart of the design method commonly called the pole-placement or pole-assignment technique, in which coherent observers are used to achieve a desired pole-placement of the quantum plant-observer system.

We explore whether an observer-based quantum controller exists, and if so how it can be designed to achieve the control goals. It will be shown that if the system considered is completely controllable, then poles of the closed-loop system may be placed at any desired locations by means of coherent feedback through an appropriate feedback gain matrix, with the assistance of a coherent quantum observer. 

By merely using field coupling, a coherent observer-based feedback controller,  has equations of the following form
\begin{subequations}
\label{eq:obs_con}
\begin{align}
d\hat{x}\left(t\right)&=F\hat{x}\left(t\right)dt+G_{1}dy\left(t\right)+G_{2}dz\left(t\right)+G_{3}dv\left(t\right),\\
du\left(t\right)&=H\hat{x}\left(t\right)dt+dz\left(t\right),
\end{align}
\end{subequations}
where the $n_x$-dimensional column vector $\hat{x}\left(t\right)$ denotes the \lq\lq{estimate}\rq\rq \  of $x\left(t\right)$. Also, matrices $G_1 \in \mathbb{R}^{n_x\times n_y}$, $G_2 \in \mathbb{R}^{n_x\times n_z}$ and $G_3 \in \mathbb{R}^{n_x\times n_v}$ respectively ($n_z$ and $n_v$ are positive even numbers). The quantum input signal $z\left(t\right)$ contributes to the controller output signal $u\left(t\right)$. Unlike the classical case, here an additional quantum noise signal $v\left(t\right)$ may be needed to ensure \eqref{eq:obs_con} correspond to a quantum physical system. The quantum noises input to \eqref{eq:obs_con} that do not contribute to $u\left(t\right)$ are also included to $v\left(t\right)$. 

Note that the system (\ref{eq:obs_con}) must satisfy the following physical realisability conditions \cite{MJ12}
\begin{subequations}
\label{eq:OCPR}
\begin{align}
&F\Theta_{n_{x}}+\Theta_{n_{x}}F^{T}+G_1\Theta_{n_{w}}G_1^{T}\nonumber\\
&+G_2\Theta_{n_z}G_2^{T} + G_3\Theta_{n_v}G_3^{T}=0,\\
&G_2=\Theta_{n_{x}}H^{T}\Theta_{n_z}.
\end{align}
\end{subequations}
The $\mathbb{R}^{n_x\times n_x}$ matrix $F$ is given by 
\begin{align}
F = A - G_1C + B_2H.
\label{eq:formF}
\end{align}
To see why $F$ is of the form as \eqref{eq:formF} shows, we define the observer error as $e\left(t\right)=x\left(t\right)-\hat{x}\left(t\right)$,  and then the mean value of $e\left(t\right)$ evolves as 
\begin{align}
\langle \dot{e}\left(t\right)\rangle=\left(A-G_1C\right)\langle e\left(t\right)\rangle
\label{eq:err}
\end{align}
by comparing (\ref{eq:qsys}) and (\ref{eq:obs_con}). 
Therefore, $\langle e\left(t\right)\rangle$ will converge to zero asymptotically if and only if $A-G_1C$ is Hurwitz. That is, in the sense of mean values, the system \eqref{eq:obs_con} can be used to track the  quantum plant \eqref{eq:qsys} by making appropriate choice of the observer gain $G_1$, which implies the controller \eqref{eq:obs_con} preserves the key property as a coherent observer.

The corresponding doubled-up form of the closed-loop system is 
\begin{align*}
d\left[\begin{array}{c}
x\left(t\right)\\
\hat{x}\left(t\right)
\end{array}\right]=\bar{\mathcal{A}}_s\left[\begin{array}{c}
x\left(t\right)\\
\hat{x}\left(t\right)
\end{array}\right]dt+\bar{\mathcal{B}}_s\left[\begin{array}{c}
dw\left(t\right)\\
dz\left(t\right)\\
dv\left(t\right)
\end{array}\right]
\end{align*}
with
\begin{align*}
\bar{\mathcal{A}}_s&=\left[\begin{array}{cc}
A & B_{2}H\\
G_{1}C & A-G_{1}C+B_{2}H
\end{array}\right],\\
\bar{\mathcal{B}}_s&=\left[\begin{array}{ccc}
B_{1} & B_{2} & 0\\
G_{1} & G_{2} & G_{3}
\end{array}\right].
\end{align*}

In terms of quantum harmonic oscillators, distinct from classical systems, one can allow for direct coupling between the plant \eqref{eq:qsys} and the observer-based feedback controller \eqref{eq:obs_con}, that is, an interaction Hamiltonian
\begin{align}
\mathcal{H}_{c}=\frac{1}{2}x^{T}R_{c}\hat{x}+\frac{1}{2}\hat{x}^{T}R_{c}^{T}x
\label{eq:interH}
\end{align}
can be introduced with $R_c \in \mathbb{R}^{n_x \times n_x}$.

In the presence of direct coupling, the combined plant and controller feedback system is described by
\begin{align}
d\left[\begin{array}{c}
x\left(t\right)\\
\hat{x}\left(t\right)
\end{array}\right]=\mathcal{A}_s\left[\begin{array}{c}
x\left(t\right)\\
\hat{x}\left(t\right)
\end{array}\right]dt+\mathcal{B}_s\left[\begin{array}{c}
dw\left(t\right)\\
dz\left(t\right)\\
dv\left(t\right)
\end{array}\right]
\label{eq:comqoc}
\end{align}
with (e.g., see \cite{IRP14,JNP08})
\begin{align}
\mathcal{A}_s=&\bar{\mathcal{A}}_s+ 2\left[\begin{array}{cc}
\Theta_{x} & 0_{n_{x}}\\
0_{n_{x}} & \Theta_{x}
\end{array}\right]\left[\begin{array}{cc}
0_{n_x} & R_c\\
R_c^T & 0_{n_x}
\end{array}\right]\nonumber\\
=&\left[\begin{array}{cc}
A & B_{2}H+2\Theta_{n_{x}}R_{c}\\
G_{1}C+2\Theta_{n_{x}}R_{c}^{T} & A-G_{1}C+B_{2}H
\end{array}\right], \label{eq:comA}\\
\mathcal{B}_s =& \bar{\mathcal{B}}_s.
\end{align}

Now we present a lemma which gives a sufficient and necessary condition for independently designing the observer and feedback, in the presence of both filed and direct couplings.
\begin{lemma}
The closed-loop poles are the combination of the poles from the observer and the poles that would have resulted from using the same feedback on the true states, if and only if $R_c$ in \eqref{eq:interH} is symmetric. Specifically, the closed-loop polynomial $\Delta_{cl}(s)$ is given by
\begin{align}
\Delta_{cl}(s) &= \mathrm{det}\left(sI_{n_{x}}-A-B_{2}H-2\Theta_{n_{x}}R_{c}\right)\nonumber\\
&\times \mathrm{det}\left(sI_{n_{x}}-A+G_{1}C.+2\Theta_{n_{x}}R_{c}\right).
\label{eq: clpoly}
\end{align}
\end{lemma}
\begin{proof}
First, by \eqref{eq:comqoc} and \eqref{eq:comA}, we know that
\begin{align*}
\left\langle \dot{x}\left(t\right)\right\rangle =&A\left\langle x\left(t\right)\right\rangle+\left(B_2H+2\Theta_{n_{x}}R_{c}\right)\left\langle \hat{x}\left(t\right)\right\rangle, \\
\left\langle \dot{\hat{x} }\left(t\right)\right\rangle =&\left(G_{1}C+2\Theta_{n_{x}}R_{c}^{T}\right)\left\langle x\left(t\right) \right\rangle \\
&+\left(A-G_{1}C+B_2H\right)\left\langle \hat{x}\left(t\right) \right\rangle ,
\end{align*}
and thus the mean error dynamics are given by
\begin{align*}
\left\langle \dot{e}\left(t\right)  \right\rangle =&\left(A-G_{1}C-2\Theta_{n_{x}}R_{c}^{T}\right)\left\langle x \left(t\right)\right\rangle \nonumber\\
&-\left(A-G_{1}C-2\Theta_{n_{x}}R_{c}\right)\left\langle \hat{x}\left(t\right)\right\rangle .
\end{align*}
Then, with the inclusion of direct coupling, the doubled-up form concerning the error is described by
\begin{align*}
d\left[\begin{array}{c}
x\left(t\right) \\
e\left(t\right)
\end{array}\right]=\mathcal{A}_e\left[\begin{array}{c}
x\left(t\right) \\
e \left(t\right) 
\end{array}\right]dt+\mathcal{B}_e\left[\begin{array}{c}
dw\left(t\right)  \\
dz\left(t\right)  \\
dv \left(t\right)
\end{array}\right]
\end{align*}
with
\begin{align*}
\mathcal{A}_e &= \left[\begin{array}{cc}
A+B_{2}H+2\Theta_{n_{x}}R_{c} & -B_{2}H - 2\Theta_{n_{x}}R_{c}\\
2\Theta_{n_x}\left(R_c - R_c^T\right) & A-G_{1}C-2\Theta_{n_{x}}R_{c}
\end{array}\right],\\
\mathcal{B}_e &= \bar{\mathcal{B}}_e.
\end{align*}
Therefore, if and only if (note that $\Theta_{n_x}^{-1} = -\Theta_{n_x}$)
\begin{align*}
R_c = R_c^T,
\end{align*}
i.e.,
\begin{align}
\mathcal{A}_e = \left[\begin{array}{cc}
A+B_{2}H+2\Theta_{n_{x}}R_{c} & -B_{2}H- 2\Theta_{n_{x}}R_{c}\\
0 & A-G_{1}C-2\Theta_{n_{x}}R_{c}
\end{array}\right],
\end{align}
the closed-loop poles of the observed-based feedback control system consist of the poles due to the pole-placement design alone and the poles due to the observer design alone. This means that the pole-placement design and the observer design of \eqref{eq:comqoc} are independent of each other, and vice versa.
\end{proof}

The eigenvalues of matrix $A+B_2H+2\Theta_{n_x}R_c$ are called the regulator poles \cite{ZDG95}.
It is known that the transient response of a linear system is related to the locations of these regulator poles \cite{KO10, GGS00}.
Constraining the poles to lie in a prescribed region, by bounding the maximum overshoot, the frequency of oscillatory modes, the delay time, the rise time, and the settling time, can ensure a satisfactory transient response \cite{KO10}. 

The following theorem shows how to design an observer-based feedback controller of the form \eqref{eq:comqoc} to place the poles at the desired locations step by step.
\begin{theorem} Assume $(A - 2\Theta_{n_x}R_c, C)$ is detectable and $(A + 2\Theta_{n_x}R_c , B_2)$ is controllable with $R_c$ symmetric. Given a desired set of poles for the plant-observer system, there always exists a coherent observer-based feedback controller \eqref{eq:comqoc}, which can realise the given pole configuration.
\end{theorem}
\begin{proof}
First, because $(A - 2\Theta_{n_x}R_c, C)$ is detectable, $A-G_1C-2\Theta_{n_x}R_c$ can be made Hurwitz by appropriately choosing the matrix $G_1$.

Second, due to the controllability of  $(A + 2\Theta_{n_x}R_c , B_2)$, then the poles of $A+B_2H+2\Theta_{n_x}R_c$ can be placed at the prescribed locations  by appropriately choosing the matrix $H$. Furthermore, $G_2$ is automatically determined by \eqref{eq:OCPR} as
\begin{align*}
G_2 = \Theta_{n_x}H^T\Theta_{n_z}.
\end{align*} 

Note that so far $G_1$, $G_2$, $H$ are determined, and $F$ is determined according to \eqref{eq:formF}. However, if the following algebraic constraint
\begin{align*}
F\Theta_{n_{x}}+\Theta_{n_{x}}F^{T}+G_1\Theta_{n_{w}}G_1^{T}+G_2\Theta_{n_z}G_2^{T}=0
\end{align*}
is not satisfied, the controller \eqref{eq:comqoc}  cannot be physically realised \cite{JNP08,MJ12}. Therefore, we inject additional quantum noise $v\left(t\right)$ to the controller to guarantee that the physical realisability conditions \eqref{eq:OCPR} hold. To be specific, if the coupling operator corresponding to $v\left(t\right)$ is denoted $L_v = \Lambda_v \hat{x}$, then $\Lambda_v$ is any $\frac{n_{v}}{2}\times n_{x}$ complex matrix satisfying
\begin{align}
\Lambda_v^{\dagger}\Lambda_v=&-\frac{i}{4}\Theta_{n_x}F-\frac{i}{4}F^{T}\Theta_{n_{x}}\nonumber\\
&+\frac{i}{4}\Theta_{n_{x}}G_1\Theta_{n_{w}}G_1^{T}\Theta_{n_{x}}\nonumber\\
&+\frac{i}{4}\Theta_{n_{x}}G_2\Theta_{n_z}G_2^{T}\Theta_{n_{x}}+\Xi_v
\label{eq:lv}
\end{align}
where $\Xi_v$ is any real symmetric $n_{x}\times n_{x}$ matrix such that the right hand side of (\ref{eq:lv}) is nonnegative definite.

Finally, $G_3$ can be explicitly given by (e.g., see \cite{JNP08,MJ12,MEPUJ13})
\begin{align}
G_3=2i\Theta_{n_{x}}\left[\begin{array}{cc}
-\Lambda_v^{\dagger} & \Lambda_v^{T}\end{array}\right]\Gamma_{n_v}.
\label{eq:G3form}
\end{align}
where $\Lambda_v$ is determined by \eqref{eq:lv}.
Here 
\begin{align*}
\Gamma_{n_v}=P_{n_v} I_{\frac{n_v}{2}} \otimes M,
\end{align*}
and $P_{n_v}$ is an $n_v\times n_v$  permutation matrix so that if we consider a column vector $a=\left[\begin{array}{cccc}a_{1} & a_{2} & \cdots & a_{n_v}\end{array}\right]^{T}$,
\begin{align*}
&P_{n_v}a\\
&=\left[\begin{array}{cccccccc}a_{1} & a_{3} & \cdots & a_{n_v-1} & a_{2} & a_{4} & \cdots & a_{n_v}\end{array}\right]^{T}.
\end{align*}
Therefore, a coherent observer-based feedback controller describe by \eqref{eq:comqoc} which solves the pole-placement problem can always be designed following the above procedure, under the assumptions given in this theorem. 
\end{proof}
\begin{remark}
If neither $\left(A,C\right)$ is detectable nor $\left(A,B_2\right)$ is controllable, then by including appropriate direct coupling between the plant and controller, the detectability and controllability assumptions in Theorem 1 can be satisfied. Furthermore, when the controller is required to be of a specific structure (e.g. a simple cavity), a greater degree of freedom is obtained for the pole-placement design by involving direct coupling so that the poles can be placed as desired.
\end{remark}

\section{Numerical examples}
\label{sec:IE}
For a second-order system with poles $z = -\zeta \omega_n \pm i\omega_d$, the transient response is characterised by the undamped natural frequency $\omega_n$, the damping ratio $\zeta$ and the damped natural frequency $\omega_d$ \cite{KO10}. In terms of quantum optical devices, these parameters are related to physical quantities such as resonant frequencies and dissipation rates. The dynamics of the quantum plant considered here are
\begin{subequations}
\label{eq:exp}
\begin{align}
dx=&\left[\begin{array}{cc}
0 & \Delta\\
-\Delta & 0
\end{array}\right]xdt+ \left[\begin{array}{cc}
0 & 0\\
0 & -2\sqrt{\kappa_{1}}
\end{array}\right]dw_1\nonumber\\
+&\left[\begin{array}{cc}
0 & 0\\
0 & -2\sqrt{\kappa_{2}}
\end{array}\right]du,\\
dy=&\left[\begin{array}{cc}
2\sqrt{\kappa_{1}} & 0\\
0 & 0
\end{array}\right]xdt+dw_{1}
\end{align}
\end{subequations}
with the detuning $\Delta = 0.1$ and the decay rates $\kappa_1 = \kappa_2 = 0.01$. This plant may be thought of as representing the scenario of an atom trapped between two mirrors of a cavity in the strong coupling limit in which the cavity dynamics are adiabatically eliminated \cite{Doherty:1999, NJP09, HLWHM11}.

1) In the following pole-placement design problem, we confine the closed-loop poles to the predetermined region specified by (as depicted by the shaded area in Fig. \ref{fig:ex}) a minimum decay rate $\alpha$, a minimum damping ratio $\zeta=\cos\theta$, and a maximum damped natural frequency $\omega_d = r\sin\theta$. 
\begin{figure}[!htp]
\centering
\includegraphics[scale=.5]{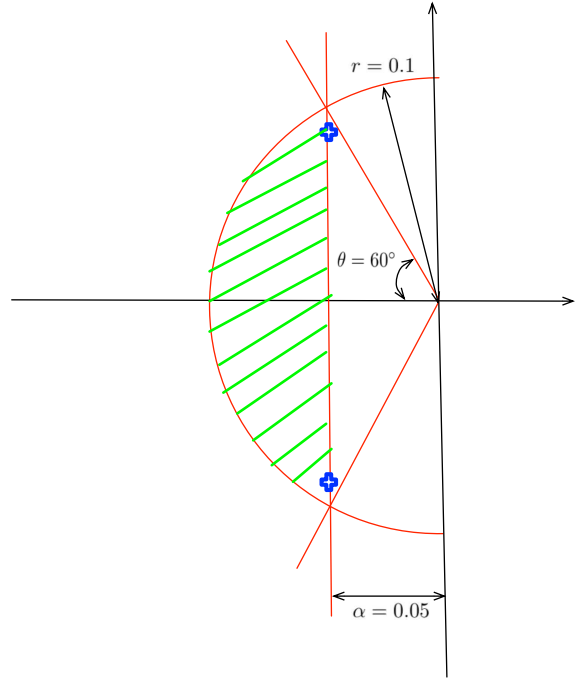}
\caption{The prescribed region specified by $r \leqslant 0.1$, $\alpha \geqslant 0.05$ and $\theta \leqslant 60^{\circ}$. Only by allowing for both direct and indirect couplings can the closed-loop poles be placed at the blue points.}
\label{fig:ex}
\end{figure} 

In order to achieve a simple structure of the observer-based coherent feedback controller, we further require the coupling operators are proportional to the annihilation operator of the controller. Hence, we have $G_1 = g_1I_2$, $H = hI_2$. Let $R_{c}=\left[\begin{array}{cc}
0 & r_{c}\\
r_{c} & 0
\end{array}\right]$. Here, $g_1$, $h$ and $r_c$ are all real numbers. Note that the quantum plant \eqref{eq:exp} is controllable and detectable.

In the absence of direct coupling, by solving the characteristic equation
\begin{align*}
\mathrm{det}\left(zI_{2}-\left(A+B_{2}H\right)\right)=z^{2}+0.2hz+0.01=0,
\end{align*} 
one can obtain the closed-loop poles as
\begin{align*}
z = -0.1h \pm i0.1\sqrt{1-h^2}.
\end{align*}
According to the design specifications, we conclude  $ 0.5  \leqslant h \leqslant 1$. But because $|z| \equiv 0.1$, only the arc in Fig. \ref{fig:ex} but not the whole shaded area can be covered.

2) Now we consider the case where direct coupling between the plant and the controller are allowed, by involving $R_c$ in the coupled dynamics. The characteristic equation becomes
\begin{align*}
&\mathrm{det}\left(zI_{2}-\left(A+B_{2}H+2JR_c\right)\right)\\
&=z^{2}+0.2hz+0.01-0.4hr-4r^2=0.
\end{align*}
We set $h = 0.5$ and $r_c = 0.01$, then the closed-loop poles are
\begin{align*}
z = -0.05 \pm i0.0714
\end{align*}
corresponding to the blue points in Fig. \ref{fig:ex}. In fact, the whole prescribed region is reachable in the presence of direct coupling. One may choose $g_1 = 1$, and the poles of the observer are $-0.166, -0.034$. Then by \eqref{eq:formF},
\begin{align*}
F = A + B_2H - G_1C = \left[\begin{array}{cc}
-0.2 & 0.1\\
-0.1 & -0.1
\end{array}\right].
\end{align*}
Finally, $G_3$ is determined by the physical realisability conditions \eqref{eq:OCPR} as $G_3 = \left[\begin{array}{cc}
-1.9 & 0\\
0 & -0.5
\end{array}\right]$. All the coefficients of the observer-based coherent feedback controller have been determined with the regulator poles placed at the desired locations. 



%

\section{Conclusions}
\label{sec:C}
A systematic approach to the pole-placement problem is provided for linear quantum stochastic systems in this paper, using coherent observers via field and direct couplings. Theorem 1 shows, in this framework, how the observer and feedback can be designed separately and combined to form the observed-state feedback control system step by step. We demonstrate in the context of examples that, by placing the poles of the closed-loop system to the prescribed regions, the transient response performance is improved. Future work includes applying the technique proposed here to quantum squeezing of motion in optomechanical systems, and conducting experiments accordingly.

\bibliographystyle{plain}
\bibliography{CDC2015}
\end{document}